\def\myjournal{pre}
\def\mystyle{reprint}

\documentclass[aps,\myjournal,\mystyle,superscriptaddress]{revtex4-2}

\usepackage{graphicx}
\usepackage{amsmath}
\usepackage{amssymb}
\usepackage{hyperref}
\usepackage{tikz}
\usetikzlibrary{intersections}
\usetikzlibrary {datavisualization.formats.functions}
\usepackage{amsthm}
\newtheorem*{theorem}{Theorem}
\newtheorem*{remark}{Remark}
\newtheorem{proposition}{Proposition}
\newtheorem{corollary}{Corollary}

\begin{document}

\title{Investigating Disordered Granular Matter via Ordered Geometric Fragmentation}

\author{Malkhazi A. Meladze}
\affiliation{Independent Researcher, Auckland, New Zealand}
\email{Contact author: malkhaz.meladze@gmail.com}

%\date{\today}

\begin{abstract}
The evolution of occupied volume under progressive fragmentation of granular matter is studied
using a purely geometric model. Rather than modelling disorder directly, properties are investigated
by analysing highly ordered reference configurations that provide sharp upper bounds on accessible
volume. Grains are idealised as fragments from a hypothetical elongated parent prism with square
cross section, sequentially sliced and reassembled into configurations that maximise enclosed volume.
Analytic expressions are derived for the maximal volume at each fragmentation stage. Volume
evolution is non-monotonic: initial fragmentation produces structures exceeding the original volume,
while further fragmentation leads to monotonic decrease converging to 5/4 times the initial volume,
independent of fragment number. The packing fraction obeys the asymptotic scaling law of inverse
proportionality to aspect ratio, in agreement with experimental observations. The model reveals pairs
of configurations built from geometrically indistinguishable building blocks yet enclosing different
volumes. These conjugate configurations constitute geometric analogues of distinct phases connected
by rearrangement-induced transitions. A criterion for observability is derived, showing such transitions
are restricted to systems of limited grain number but may occur locally as domain formation in larger
assemblies. Comparison with experimental data confirms the model provides a lower bound on packing
fraction and predicts domain sizes should scale linearly with aspect ratio, testable through X-ray tomography.
\end{abstract}

\maketitle

\section{Introduction}
\label{sec:intro}

Packing problems arise in a wide range of physical contexts, including granular compaction, powder processing,
jamming phenomena, and colloidal self-assembly. A substantial body of work has therefore been devoted to
understanding how particle shape, anisotropy, and geometric constraints influence packing density and
available volume; see, for example, Refs.~\cite{mehta2007,aranson2009,torquato2018} and references therein.
In most studies, the primary objective is to determine optimal or typical packing fractions for prescribed
particle shapes within disordered or statistically defined ensembles.

Recent theoretical progress has been made in understanding random close packing (RCP) of monodisperse
spherical particles through analytical frameworks linking packing fraction to coordination number and the onset
of mechanical rigidity~\cite{zaccone2025}. These approaches identify the jamming transition as a critical point
characterized by a threshold coordination number, with packing fractions that are largely independent
of system size. However, the case of non-spherical elongated particles remains analytically challenging and has
been explicitly identified as an open problem requiring further investigation~\cite{zaccone2025}. The present
work contributes to this area through a purely geometric approach that reveals how fragmentation-induced volume
evolution in elongated grain assemblies can be understood from geometric constraints alone, without recourse to
mechanical or energetic considerations.

In contrast to conventional packing studies, the present work addresses a complementary question: how does the
volume occupied by granular matter evolve under progressive fragmentation? Rather than modelling disorder
directly, this question is examined through a controlled geometric construction based on a single elongated
parent object that is sequentially partitioned into identical building blocks. At each fragmentation stage,
the maximal enclosed volume attainable after reassembly is evaluated. The resulting configurations are not
intended to represent typical granular packings; instead, they provide sharp geometric bounds that elucidate
systematic trends in volume evolution induced by fragmentation.

The motivation for this question is partly phenomenological. Fragmentation alters particle shapes and increases
orientational degrees of freedom, which can dramatically affect packing efficiency and accessible free volume.
A familiar qualitative example is that a fixed mass of finely milled powder often occupies a larger volume than
the same mass of intact grains. For very fine powders (particle sizes below approximately $100\,\mu$m),
this effect is dominated by attractive cohesive forces such as Van der Waals interactions that hinder compaction.
The geometric effects considered in the present model apply primarily to larger grains where such cohesive forces
are negligible and geometric constraints dominate packing behaviour. Recent comprehensive reviews
\cite{zakine2025,jiao2011,torquato2018,borzsonyi2013} have emphasized the central role of particle anisotropy in
controlling packing behaviour in granular systems. The present work was further motivated by the
geometric perspective introduced in Ref.~\cite{zakine2025}.

The packing of elongated particles has been studied extensively both experimentally and
theoretically~\cite{borzsonyi2013,philipse1996,freeman2019,villarruel2000}.
Systematic experiments have demonstrated that packing fraction decreases with increasing aspect ratio for
randomly packed rods~\cite{freeman2019,borzsonyi2016}, with values ranging from $\phi \approx 0.57$ for
short rods (aspect ratio $\alpha \approx 4$) to $\phi \approx 0.13$ for highly elongated particles
($\alpha \approx 32$). Importantly, experimental measurements reveal that nominally identical particles under
similar preparation conditions can yield a \emph{range} of packing fractions rather than a single reproducible
value~\cite{villarruel2000}, indicating the presence of multiple metastable configurations\,---\,a phenomenon
that finds a natural geometric interpretation in the present model.

While much attention has been devoted to the packing of particles with fixed shapes, the influence of
fragmentation on packing density remains less systematically explored. Fragmentation processes\,---\,whether
mechanical comminution, crushing, or controlled slicing\,---\,alter both the size distribution and the
geometric properties of constituent grains~\cite{cho2006}. The resulting changes in particle aspect ratio
and surface morphology can significantly modify packing efficiency. The present work isolates the purely
geometric contribution to these effects by considering idealized fragmentation of a single parent object into
progressively smaller, identical building blocks.

Building on these ideas, a minimal, fully analytic model is constructed that isolates the geometric consequences
of fragmentation and reassembly. Explicit expressions for the maximal enclosed volume at each fragmentation
stage are obtained, and several structural features are revealed. In particular, the emergence of paired
(conjugate) configurations composed of geometrically indistinguishable building blocks, yet enclosing
different volumes, is identified. The model predicts bounds on packing density that are shown to be consistent
with experimental scaling laws for elongated particles. Furthermore, the geometric framework reveals local
ordering phenomena in bulk granular assemblies that can be tested through X-ray tomography experiments.

The objectives of this paper are as follows: (i) to quantify the evolution of volume under sequential
fragmentation of granular matter using a simple geometric model; (ii) to investigate structural relations
among configurations reassembled from differently fragmented building blocks; (iii) to derive explicit
formulae that connect the parameters of the model to experimentally measurable quantities; and (iv) to
compare theoretical predictions with experimental data and propose testable signatures of the geometric
mechanisms identified here.

\section{Model and notation}
\label{sec:model}
In this work, the geometric properties of granular matter composed of idealised grains are investigated
within a purely static geometric framework. Friction between grains is assumed to be negligible, and the
weight of individual grains is ignored. All grains are assumed to originate from a single, very long
hypothetical parent object. The number of grains produced and their individual sizes are determined entirely
by the geometry of this parent object.

Grains are modelled as rectangular parallelepipeds with square cross-section. Accordingly, a family of long
rectangular prisms (hereafter referred to as parent objects) is considered, denoted by $T_{\,n}^{\,0}$,
with square cross-section $a\times a$ and length
\begin{equation}
    l = 2^{\,n+2} a, \qquad n = 0,1,2,\dots .
    \label{eq:Init_length_model}
\end{equation}
The length is written in this particular form for convenience in subsequent calculations: the initial
fragmentation divides the length by four, reducing the exponent by two. While representation of $l$ as
a power of two restricts the length of the parent object to even multiples of $a$, this restriction does not 
affect applicability for large grain numbers, where variations of a single grain are negligible.

Each parent object may conveniently be viewed as consisting of
\begin{equation}
N_{\rm tot} = 2^{\,n+2}
\label{eq:Ntot}
\end{equation}
unit cubes of side length $a$. The corresponding volume is
\begin{equation}
V_{\,n}^{\,0} = 2^{\,n+2} a^{3}.
\label{eq:V0}
\end{equation}

Fragmentation of the parent object is carried out stage by stage. At each stage $i=1,2,\dots,n+1$,
the object is fragmented according to a rule defined below, and the resulting pieces (grains) are
reassembled into a configuration that maximises the volume. The reassembled object is referred to as
a \emph{tower} and is denoted by $T_{\,n}^{\,i}$\,, with volume $V_{\,n}^{\,i}$ and relative volume
\begin{equation}
R_{\,n}^{\,i} = \frac{V_{\,n}^{\,i}}{V_{\,n}^{\,0}}.
\label{eq:RelVol}
\end{equation}
In this notation, the lower index $n$ characterises the size of the parent
object via Eq.~\eqref{eq:Init_length_model}, while the upper index $i$ labels the
fragmentation stage.

\subsection{Fragmentation rule}
\label{sec:fragmntrule}

The fragmentation rule is defined recursively as follows. Stage~0 denotes the initial configuration
prior to any fragmentation, while stages $i = 1, 2, \dots, n+1$ correspond to successive fragmentation steps.

\begin{enumerate}
    \item \textbf{Stage 0:} The parent object $T_{\,n}^{\,0}$ is defined
          as a rectangular prism of dimensions $2^{\,n+2}a \times a \times a$.
    \item \textbf{Stage 1:} The prism is divided into four equal parts of
          length $l/4 = 2^{\,n}a$. These parts are reassembled to maximise
          the enclosed volume, yielding the tower $T_{\,n}^{\,1}$.
    \item \textbf{Stage 2:} Each piece obtained at Stage~1 is divided along
          its longest axis into two equal parts. The resulting eight pieces
          are reassembled into the tower $T_{\,n}^{\,2}$.
    \item This procedure is continued by dividing each piece along its longest
          axis into two equal parts at each subsequent stage, until stage $n+1$,
          at which point all pieces are unit cubes and the maximal-volume
          configuration $T_{\,n}^{\,n+1}$ is obtained.
\end{enumerate}

After fragmentation, the elementary building blocks can, in general, be assembled into towers whose cross
sections form parallelograms. In the present model, the cross section is restricted to be square. This choice
is motivated by the fact that, among all parallelograms with a given perimeter, the square encloses
the maximal cross-sectional area.

The square geometry therefore represents an extremal, volume-maximising configuration within this class
and serves as a natural reference case. No claim is made that this construction exhausts all possible
configurations, nor that a global maximal-volume theorem is established.

\section{Fragmentation \textendash{} reassembly process and volume evolution}
\label{sec:volevolution}

Using the fragmentation rule introduced in Sec.~\ref{sec:fragmntrule}, the slicing of long objects is
presented explicitly below for the first few values of $n$. For each $n$, the sequence of
towers $T_{\,n}^{\,i}$ produced at fragmentation stages $i = 1, \dots, n+1$, including the parent
object $i=0$, is listed alongside their volumes $V_{\,n}^{\,i}$ and corresponding relative volumes
$R_{\,n}^{\,i}$. These examples motivate the general formulae, which are subsequently 
derived in the following subsections.

In the figures that follow, thick lines indicate grain boundaries, while thin lines serve only to
illustrate the relative sizes of the grains.

\subsection{\texorpdfstring{Case $n=0$ ($l = 4a$)}{Case n=0 (l=4a)}}

The present case is the simplest and is illustrated in Fig.~\ref{fig:n0}. The parent object contains
$2^{0+2}=4$ unit cubes and has volume
\[
V_{\,0}^{\,0} = 4a^{3}.
\]
Only one fragmentation stage is possible. The prism $T_{\,0}^{\,0}$ is divided into four equal parts
and reassembled into a cross-shaped tower $T_{\,0}^{\,1}$ with a cavity at its centre, as shown
in Fig.~\ref{fig:n0}. The resulting volume is
\[
V_{\,0}^{\,1} = 5a^{3},
\qquad
R_{\,0}^{\,1} = \frac{5}{4}.
\]
Thus, the maximal-volume configuration obtained by reassembling the fragments occupies a volume larger
than that of the original object by a factor $5/4$.

The four cubes can also be rearranged into a compact configuration with no gaps, reproducing the original
volume. The significance of this example is that random or unconstrained reassembly of the fragments can
lead to a volume larger than that of the original object, while remaining bounded above by the maximal
value $5/4$.

\begin{figure}[t]
\centering
\begin{tikzpicture}[scale=0.4]
\draw[thick] (0,0) node[below]{$a$} node[above=10pt, left]{$a$}--(0,1)--(0.5,0.4)--(0.5,-0.6)--cycle;
\foreach \x in {1,2,3} \draw (\x, 1)--(\x+0.5, 0.4)--(\x+0.5,-0.6);
\draw[thick] (4,1)--(4.5,0.4)--(4.5,-0.6);
\draw[thick] (0,1)--(2,1) node[above] {$l=4a$}--(4,1);
\draw[thick] (0.5,0.4)--(4.5,0.4);
\draw[thick] (0.5,-0.6)--(2.5,-0.6) node[below=7pt]{$T_{\,0}^{\,0}$}--(4.5,-0.6);
\draw[thick] (12.5,0.4)--(13.5,0.4)--(14,-0.2)--(13,-0.2)--cycle;
\draw[thick] (14,-0.2)--(14,-1.2)--(13,-1.2)--(12.5,-0.6)--(12.5,0.4);
\draw[thick] (13,-1.2) node[below]{$T_{\,0}^{\,1}$}--(13,-0.2);

\draw[thick] (11,1)--(12,1)--(12.5,0.4)--(11.5,0.4)--cycle; % Top
\draw[thick] (11,1)--(11,0)--(11.5,-0.6)--(12.5,-0.6)--(12.5,0.4); % Left Front
\draw[thick] (11.5,-0.6)--(11.5,0.4);

\draw[thick] (12,1)--(11.5,1.6)--(12.5,1.6)--(13,1)--cycle;  \draw[thick] (11.5,1.6)--(11.5,1);

\draw[thick] (13,1)--(14,1)--(14.5, 0.4)--(13.5,0.4)--cycle;
\draw[thick] (14.5,0.4)--(14.5,-0.6)--(14,-0.6);
\draw[thick] (13,1)--(13,0.4);
\end{tikzpicture}
\caption{Fragmentation for case $n=0 \, (l=4a)$. The parent object $T_{\,0}^{\,0}$ and the resulting tower $T_{\,0}^{\,1}$
are shown in a cross-shaped arrangement.}
\label{fig:n0}
\end{figure}

In the following, the findings are summarised and rewritten in a form that is convenient for deriving general
formulae for the volumes, relative volumes, and the total number of fragmentation stages. $T(n)$ denotes the total
number of fragmentation stages for a parent object indexed by $n$, so that

\begin{equation*}
\begin{aligned}
    V_{\,0}^{\,0} &= 4a^3 = 4\times2^0 a^3, \\
    V_{\,0}^{\,1} &= 4a^3 + a^3 = 4\times2^0 a^3 + 2^{2\times0}a^3, \\
    R_{\,0}^{\,1} &= \frac{5}{4}, \\
    T(0) &= 0+1.
\end{aligned}
\label{eq:n0summary}
\end{equation*}

\subsection{\texorpdfstring{Case $n=1$ ($l = 8a$)}{Case n=1 (l=8a)}}
\label{sec:casen1}

In the case $n=1$ (see Fig.~\ref{fig:n1}), the parent object $T_{\,1}^{\,0}$ contains $2^{1+2}=8$ unit cubes
and has volume
\[
	V_{\,1}^{\,0} = 8a^{3}.
\]
After the first fragmentation, which produces four pieces of length $2^{1}a=2a$, the maximal-volume
tower $T_{\,1}^{\,1}$ has
\[
	V_{\,1}^{\,1} = 8a^{3} + 4a^{3} = 12a^{3},
	\qquad
	R_{\,1}^{\,1} = \frac{3}{2}.
\]
The quantity $4a^3$ corresponds to the volume of the cavity created at this stage.

Further division of each building block into two equal parts leads to the final stage, at which each
block is a cube. These cubes are assembled into a two-level tower $T_{\,1}^{\,2}$ with a smaller cavity of volume $2a^3$,
\[
	V_{\,1}^{\,2} = 8a^{3} + 2a^{3} = 10a^{3},
	\qquad
	R_{\,1}^{\,2} = \frac{5}{4}.
\]
The resulting sequence of relative volumes is therefore
\[
	R_{\,1}^{\,1} = \frac{3}{2} \;\longrightarrow\; R_{\,1}^{\,2} = \frac{5}{4}.
\]

\begin{figure}[t]
\begin{tikzpicture}[scale=0.4]
\draw[thick] (0,0-1) node[below]{$a$} node[above=10pt, left]{$a$}--(0,1-1)--(0.5,0.4-1)--(0.5,-0.6-1)--cycle;
\foreach \x in {1,2,3,4,5,6,7} \draw (\x, 1-1)--(\x+0.5, 0.4-1)--(\x+0.5,-0.6-1);
\draw[thick] (8,1-1)--(8.5,0.4-1)--(8.5,-0.6-1);
\draw[thick] (0,1-1)--(4,1-1) node[above] {$l=8a$}--(8,1-1);
\draw[thick] (0.5,0.4-1)--(8.5,0.4-1);
\draw[thick] (0.5,-0.6-1)--(4.5,-0.6-1) node[below=7pt]{\(T_{\,1}^{\,0}\)}--(8.5,-0.6-1);

\begin{scope}[shift={(-1,0)}]
	\draw (14,-0.2-1)--(14.5,-0.8-1)--(14.5,-1.8-1)node[below]{\(T_{\,1}^{\,1}\)};
	\draw[thick] (13,-0.2-1)--(15,-0.2-1);
	\draw[thick] (13.5,-0.8-1)--(15.5,-0.8-1);
	\draw[thick] (13,-0.2-1)--(13.5,-0.8-1); \draw[thick] (15,-0.2-1)--(15.5,-0.8-1);
	\draw[thick] (13,-1.2-1)--(13.5,-1.8-1)--(15.5,-1.8-1)--(15.5,-0.8-1);
	\draw[thick] (13.5,-1.8-1)--(13.5,-0.8-1);

	\draw (12.5,0.4-1)--(11.5,0.4-1)--(11.5,-0.6-1);
	\draw[thick] (12,1-1)--(13,-0.2-1); \draw[thick] (11,1-1)--(12,-0.2-1);% Top
	\draw[thick] (11,1-1)--(12,1-1); \draw[thick] (12,-0.2-1)--(13,-0.2-1);
	\draw[thick] (11,0-1)--(12,-1.2-1);
	\draw[thick] (11,1-1)--(11,0-1);  \draw[thick] (12,-0.2-1)--(12,-1.2-1);
	\draw[thick] (12,-1.2-1)--(13,-1.2-1)--(13,-0.2-1);

	\draw (12.5,1.6-1)--(13,1-1)--(13,0-1);
	\draw[thick] (12,1-1)--(14,1-1); \draw[thick] (11.5,1.6-1)--(13.5,1.6-1);
	\draw[thick] (11.5,1.6-1)--(12,1-1); \draw[thick] (13.5,1.6-1)--(14,1-1);

	\path [name path=vline] (11.5,1.6-1)--(11.5,0.6-1);
	\path [name path=hline] (11,1-1)--(12,1-1);
	\draw [name intersections={of=hline and vline, by=x}] [thick] (11.5,1.6-1)--(x);

	\path [name path=hline] (14,0-1)--(12,0-1);
	\path [name path=slope] (12,1-1)--(13,-0.2-1);
	\draw [name intersections={of=hline and slope, by=x}] [thick] (14,0-1)--(x);

	\draw (15.5,0.4-1)--(14.5,0.4-1)--(14.5,-0.2-1);

	\path [name path=hline] (13,-0.2-1)--(15,-0.2-1);
	\path [name path=slope] (14,0-1)--(15,-0.8-1);
	\draw [name intersections={of=hline and slope, by=x}][thick] (14,0-1)--(x);

	\draw[thick] (14,1-1)--(15,-0.2-1); \draw[thick] (15,1-1)--(16,-0.2-1);
	\draw[thick] (14,1-1)--(15,1-1); \draw[thick] (15,-0.2-1)--(16,-0.2-1);
	\draw[thick] (16,-0.2-1)--(16,-1.2-1)--(15.5,-1.2-1);
	\draw[thick] (14,1-1)--(14,0-1);
\end{scope}

\begin{scope}[shift={(-2,0)}]
	\draw[thick] (12.5+7,0.4)--(13.5+7,0.4)--(14+7,-0.2)--(13+7,-0.2)--cycle;
	\draw[thick] (14+7,-0.2)--(14+7,-1.2)--(13+7,-1.2)--(12.5+7,-0.6)--(12.5+7,0.4);
	\draw[thick] (13+7,-1.2)--(13+7,-0.2);

	\draw[thick] (11+7,1)--(12+7,1)--(12.5+7,0.4)--(11.5+7,0.4)--cycle; % Top
	\draw[thick] (11+7,1)--(11+7,0)--(11.5+7,-0.6)--(12.5+7,-0.6)--(12.5+7,0.4); % Left Front
	\draw[thick] (11.5+7,-0.6)--(11.5+7,0.4);

	\draw[thick] (12+7,1)--(11.5+7,1.6)--(12.5+7,1.6)--(13+7,1)--cycle;  \draw[thick] (11.5+7,1.6)--(11.5+7,1);

	\draw[thick] (13+7,1)--(14+7,1)--(14.5+7, 0.4)--(13.5+7,0.4)--cycle;
	\draw[thick] (14.5+7,0.4)--(14.5+7,-0.6)--(14+7,-0.6);
	\draw[thick] (13+7,1)--(13+7,0.4);

	\draw[thick] (18,0)--(18,-1)--(18.5,-1.6)--(19.5,-1.6)--(20,-2.2)node[below=6pt]{\(T_{\,1}^{\,2}\)}--(21,-2.2);
	\draw[thick] (18.5,-0.6)--(18.5,-1.6); \draw[thick] (19.5,-0.6)--(19.5,-1.6);
	\draw[thick] (20,-2.2)--(20,-1.2); \draw[thick] (21,-2.2)--(21,-1.2);
	\draw[thick] (14.5+7,-0.6)--(14.5+7,-1.6)--(21,-1.6);
\end{scope}

\end{tikzpicture}
\caption{Fragmentation stages for case $n=1 \, (l=8a)$. The parent object $T_{\,1}^{\,0}$ and the two resulting towers
$T_{\,1}^{\,1}$ and $T_{\,1}^{\,2}$ are shown. The towers form a conjugate pair, with the second tower enclosing a smaller volume.}
\label{fig:n1}
\end{figure}

In the following, the findings are summarised in a form convenient for deriving general expressions for the volumes,
relative volumes, and the total number of fragmentation stages.

\begin{align*}
	V_{\,1}^{\,0} &= 8a^3 = 4\times 2^1 a^3,\\
	V_{\,1}^{\,1} &= 8a^3 + 4a^3 = 4\times 2^1 a^3 + 2^{2\times1}a^3,\\
	V_{\,1}^{\,2} &= 8a^3 + 2a^3 = 4\times 2^1 a^3 + 2^{2\times1-1}a^3,\\
	T(1) &= 1+1.
\end{align*}

The tower $T_{\,1}^{\,1}$ is built from blocks of length $2a$, whereas the tower $T_{\,1}^{\,2}$ is constructed
from unit blocks arranged such that the resulting structure appears, at a coarse-grained geometric level,
as if it were assembled from the same building blocks as $T_{\,1}^{\,1}$, differing only in their global arrangement.
In this sense, one configuration may be transformed into the other by rearrangement of the building blocks,
despite differences in their internal composition.

Situations of this type recur in all subsequent cases. It is therefore useful to introduce a specific terminology:
a pair of towers that either \emph{appear} to be, or are \emph{actually} constructed from the same building blocks,
but differ in their overall geometry and volume, are referred to as \textbf{conjugate towers}.

Within the purely geometric framework of the present model, if no energetic preference is associated with either
configuration, both may occur with comparable probability despite their different volumes. In this sense,
the two configurations may be viewed as distinct `geometric phases', one denser than the other.

\subsection{\texorpdfstring{Case $n=2$ ($l = 16a$)}{Case n=2 (l=16a)}}
\label{sec:casen2}

In the present case (see Fig.~\ref{fig:n2}), the parent object $T_{\,2}^{\,0}$ consists of $16$ unit cubes
arranged in a single row.

The first tower $T_{\,2}^{\,1}$ is obtained by dividing the parent object into four identical parts,
each of length $4a$, and arranging the resulting building blocks as shown in Fig.~\ref{fig:n2}. In this
configuration, a cavity is formed at the centre of the tower. The total volume of the tower is
\[
    V_{\,2}^{\,1} = 32a^3,
\]
where the volume of the cavity is $16a^3$.

By subsequently dividing each building block into two equal parts at each stage, two additional towers,
$T_{\,2}^{\,2}$ and $T_{\,2}^{\,3}$, are obtained. Their corresponding volumes are
\[
    V_{\,2}^{\,2} = 24a^3,  \qquad  V_{\,2}^{\,3} = 20a^3.
\]
The resulting sequence of relative volumes is therefore
\[
    R_{\,2}^{\,1} = 2 \;\longrightarrow\; 
    R_{\,2}^{\,2} = \frac{3}{2} \;\longrightarrow\; 
    R_{\,2}^{\,3} = \frac{5}{4}.
\]

\begin{figure}[t]
\begin{tikzpicture}[scale=0.4]
%\draw [help lines] (0,-10) grid [xstep = 1, ystep = 1] (30,2);
\draw[thick] (0,0) node[below]{$a$} node[above=10pt, left]{$a$}--(0,1)--(0.5,0.4)--(0.5,-0.6)--cycle;
\foreach \x in {1,2,3,4,5,6,7,8,9,10,11,12,13,14,15} \draw (\x, 1)--(\x+0.5, 0.4)--(\x+0.5,-0.6);
\draw[thick] (16,1)--(16.5,0.4)--(16.5,-0.6);
\draw[thick] (0,1)--(8,1) node[above] {$l=16a$}--(16,1);
\draw[thick] (0.5,0.4)--(16.5,0.4);
\draw[thick] (0.5,-0.6)--(8.5,-0.6) node[below]{$T_{\,2}^{\,0}$}--(16.5,-0.6);

\begin{scope}[shift={(-20,-6.3)}]
	\foreach \x in {1,2,3} \draw (22+\x, -1.4)--(22+\x+0.5, -2)--(22+\x+0.5,-3);
	\draw[thick] (22,-1.4)--(26,-1.4);
	\draw[thick] (22.5,-2)--(26.5,-2);
	\draw[thick] (22,-1.4)--(22.5,-2);  \draw[thick] (26,-1.4)--(26.5,-2);
	\draw[thick] (22,-2.4)--(22.5,-3)--(24.5,-3) node[below]{$T_{\,2}^{\,1}$}--(26.5,-3)--(26.5,-2);

	\foreach \x in {1,2,3} \draw (20+\x*0.5, 1-\x*0.6)--(19+\x*0.5, 1-\x*0.6)--(19+\x*0.5,-\x*0.6);
	\draw[thick] (20,1)--(22,-1.4); \draw[thick] (19,1)--(21,-1.4);% Top
	\draw[thick] (19,1)--(20,1); \draw[thick] (21,-1.4)--(22,-1.4);
	\draw[thick] (19,0)--(21,-2.4);
	\draw[thick] (19,1)--(19,0);  \draw[thick] (21,-1.4)--(21,-2.4);
	\draw[thick] (21,-2.4)--(22,-2.4)--(22,-1.4);
	\draw[thick] (22.5,-2)--(22.5,-3);

	\foreach \x in {1,2,3} \draw (19.5+\x, 1.6)--(19.5+\x+0.5, 1)--(19.5+\x+0.5,0);
	\draw[thick] (20,1)--(24,1); \draw[thick] (19.5,1.6)--(23.5,1.6);
	\draw[thick] (20,1)--(19.5,1.6); \draw[thick] (24,1)--(23.5,1.6);
	\draw[thick] (19.5,1)--(19.5,1.6);

	\path [name path=horizontal line] (24,0)--(19,0);
	\path [name path=slope] (20,1)--(22,-1.4);
	\draw [name intersections={of=horizontal line and slope, by=x}] [thick] (24,0)--(x);

	\foreach \x in {1,2} \draw (25+\x*0.5, 1-\x*0.6)--(24+\x*0.5, 1-\x*0.6)--(24+\x*0.5,-\x*0.6);
	\draw (25+3*0.5, 1-3*0.6)--(24+3*0.5, 1-3*0.6);

	\path [name path=hline] (22,-1.4)--(26,-1.4);
	\path [name path=vline] (24+3*0.5, 1-3*0.6)--(24+3*0.5, -3*0.6);
	\draw [name intersections={of=hline and vline, by=x}] (24+3*0.5, 1-3*0.6)--(x);

	\draw[thick] (24,1)--(26,-1.4); \draw[thick] (25,1)--(27,-1.4);
	\draw[thick] (24,1)--(25,1); \draw[thick] (26,-1.4)--(27,-1.4);
	\draw[thick] (27,-1.4)--(27,-2.4)--(26.5,-2.4);
	\draw[thick] (24,1)--(24,0);

	\path [name path=horizontal line] (22,-1.4)--(26,-1.4);
	\path [name path=slope] (24,0)--(26,-2.4);
	\draw [name intersections={of=horizontal line and slope, by=x}] [thick] (24,0)--(x);
\end{scope}

\begin{scope}[shift={(5.4,-0.5)}]
	\draw (14-8,-0.2-6)--(14.5-8,-0.8-6)--(14.5-8,-1.8-6);
	\draw[thick] (13-8,-0.2-6)--(15-8,-0.2-6);
	\draw[thick] (13.5-8,-0.8-6)--(15.5-8,-0.8-6);
	\draw[thick] (13-8,-0.2-6)--(13.5-8,-0.8-6); \draw[thick] (15-8,-0.2-6)--(15.5-8,-0.8-6);
	\draw[thick] (13-8,-1.2-6)--(13.5-8,-1.8-6)--(15.5-8,-1.8-6)--(15.5-8,-0.8-6);
	\draw[thick] (13.5-8,-1.8-6)--(13.5-8,-0.8-6);

	\draw (12.5-8,0.4-6)--(11.5-8,0.4-6)--(11.5-8,-0.6-6);
	\draw[thick] (12-8,1-6)--(13-8,-0.2-6); \draw[thick] (11-8,1-6)--(12-8,-0.2-6);% Top
	\draw[thick] (11-8,1-6)--(12-8,1-6); \draw[thick] (12-8,-0.2-6)--(13-8,-0.2-6);
	\draw[thick] (11-8,0-6)--(12-8,-1.2-6);
	\draw[thick] (11-8,1-6)--(11-8,0-6);  \draw[thick] (12-8,-0.2-6)--(12-8,-1.2-6);
	\draw[thick] (12-8,-1.2-6)--(13-8,-1.2-6)--(13-8,-0.2-6);

	\draw (12.5-8,1.6-6)--(13-8,1-6)--(13-8,0-6);
	\draw[thick] (12-8,1-6)--(14-8,1-6); \draw[thick] (11.5-8,1.6-6)--(13.5-8,1.6-6);
	\draw[thick] (11.5-8,1.6-6)--(12-8,1-6); \draw[thick] (13.5-8,1.6-6)--(14-8,1-6);

	\path [name path=vline] (11.5-8,1.6-6)--(11.5-8,0.6-6);
	\path [name path=hline] (11-8,1-6)--(12-8,1-6);
	\draw [name intersections={of=hline and vline, by=x}] [thick] (11.5-8,1.6-6)--(x);

	\path [name path=hline] (14-8,0-6)--(12-8,0-6);
	\path [name path=slope] (12-8,1-6)--(13-8,-0.2-6);
	\draw [name intersections={of=hline and slope, by=x}] [thick] (14-8,0-6)--(x);

	\draw (15.5-8,0.4-6)--(14.5-8,0.4-6)--(14.5-8,-0.2-6);

	\path [name path=hline] (13-8,-0.2-6)--(15-8,-0.2-6);
	\path [name path=slope] (14-8,0-6)--(15-8,-0.8-6);
	\draw [name intersections={of=hline and slope, by=x}][thick] (14-8,0-6)--(x);

	\draw[thick] (14-8,1-6)--(15-8,-0.2-6); \draw[thick] (15-8,1-6)--(16-8,-0.2-6);
	\draw[thick] (14-8,1-6)--(15-8,1-6); \draw[thick] (15-8,-0.2-6)--(16-8,-0.2-6);
	\draw[thick] (16-8,-0.2-6)--(16-8,-1.2-6)--(15.5-8,-1.2-6);
	\draw[thick] (14-8,1-6)--(14-8,0-6);
	\draw[thick] (16-8,-1.2-6)--(16-8,-2.2-6)--(15.5-8,-2.2-6);

	\draw[thick] (3,-6)--(3,-7)--(4,-8.2)--(5,-8.2)--(5.5,-8.8)--(7.5,-8.8)--(7.5,-7.8);
	\draw (3.5,-7.6)--(3.5,-6.6); \draw (6.5,-8.8)node[below=1pt]{$T_{\,2}^{\,2}$}--(6.5,-7.8);
	\draw[thick] (4,-8.2)--(4,-7.2);
	\draw[thick] (5,-8.2)--(5,-7.2); \draw[thick] (5.5,-8.8)--(5.5,-7.8);
	\draw[thick] (5,-6)--(5,-6.2); \draw[thick] (6,-6)--(6,-6.2);
\end{scope}

\begin{scope}[shift={(-3,1)}]
	\draw[thick] (12.5+7,0.4-6)--(13.5+7,0.4-6)--(14+7,-0.2-6)--(13+7,-0.2-6)--cycle;
	\draw[thick] (14+7,-0.2-6)--(14+7,-1.2-6)--(13+7,-1.2-6)--(12.5+7,-0.6-6)--(12.5+7,0.4-6);
	\draw[thick] (13+7,-1.2-6)--(13+7,-0.2-6);

	\draw[thick] (11+7,1-6)--(12+7,1-6)--(12.5+7,0.4-6)--(11.5+7,0.4-6)--cycle; % Top
	\draw[thick] (11+7,1-6)--(11+7,0-6)--(11.5+7,-0.6-6)--(12.5+7,-0.6-6)--(12.5+7,0.4-6); % Left Front
	\draw[thick] (11.5+7,-0.6-6)--(11.5+7,0.4-6);

	\draw[thick] (12+7,1-6)--(11.5+7,1.6-6)--(12.5+7,1.6-6)--(13+7,1-6)--cycle;  \draw[thick] (11.5+7,1.6-6)--(11.5+7,1-6);

	\draw[thick] (13+7,1-6)--(14+7,1-6)--(14.5+7, 0.4-6)--(13.5+7,0.4-6)--cycle;
	\draw[thick] (14.5+7,0.4-6)--(14.5+7,-0.6-6)--(14+7,-0.6-6);
	\draw[thick] (13+7,1-6)--(13+7,0.4-6);

	\draw[thick] (18,0-6)--(18,-1-6)--(18.5,-1.6-6)--(19.5,-1.6-6)--(20,-2.2-6)--(21,-2.2-6);
	\draw[thick] (18.5,-0.6-6)--(18.5,-1.6-6); \draw[thick] (19.5,-0.6-6)--(19.5,-1.6-6);
	\draw[thick] (20,-2.2-6)--(20,-1.2-6); \draw[thick] (21,-2.2-6)--(21,-1.2-6);
	\draw[thick] (14.5+7,-0.6-6)--(14.5+7,-1.6-6)--(21,-1.6-6);

	\draw[thick] (18,-1-6)--(18,-2-6)--(18.5,-2.6-6)--(19.5,-2.6-6)--(20,-3.2-6)--(21,-3.2-6);
	\draw[thick] (18.5,-1.6-6)--(18.5,-2.6-6); \draw[thick] (19.5,-1.6-6)--(19.5,-2.6-6);
	\draw[thick] (20,-3.2-6)--(20,-2.2-6); \draw[thick] (21,-3.2-6)--(21,-2.2-6);
	\draw[thick] (14.5+7,-1.6-6)--(14.5+7,-2.6-6)--(21,-2.6-6);

	\draw[thick] (18,-2-6)--(18,-3-6)--(18.5,-3.6-6)--(19.5,-3.6-6)--(20,-4.2-6) node[below=1pt]{$T_{\,2}^{\,3}$}--(21,-4.2-6);
	\draw[thick] (18.5,-2.6-6)--(18.5,-3.6-6); \draw[thick] (19.5,-2.6-6)--(19.5,-3.6-6);
	\draw[thick] (20,-4.2-6)--(20,-3.2-6); \draw[thick] (21,-4.2-6)--(21,-3.2-6);
	\draw[thick] (14.5+7,-2.6-6)--(14.5+7,-3.6-6)--(21,-3.6-6);
\end{scope}

\end{tikzpicture}
\caption{Fragmentation stages for case $n=2 \, (l=16a)$. The parent object $T_{\,2}^{\,0}$ and the three
resulting towers $T_{\,2}^{\,1}$, $T_{\,2}^{\,2}$ and $T_{\,2}^{\,3}$ are shown. The volumes decrease
monotonically with stage index. Towers $T_{\,2}^{\,1}$ and $T_{\,2}^{\,3}$ form a conjugate pair,
while the intermediate tower $T_{\,2}^{\,2}$ is self-conjugate.}
\label{fig:n2}
\end{figure}

Relative to the parent object, the volume of the towers first increases at the initial fragmentation stage
and then decreases monotonically at subsequent stages. This behaviour is governed by the cavity formed during
the first rearrangement: while a cavity is created at the first stage, its volume decreases as the
fragmentation proceeds.

A further parallel with the case $n=1$ (Sec.~\ref{sec:casen1}) may be drawn. The towers $T_{\,2}^{\,1}$ and
$T_{\,2}^{\,3}$ may be regarded as being built from identical building elements of length $4a$, arranged
in different geometric configurations. As before, such a pair is referred to as a pair of conjugate towers.
The intermediate tower $T_{\,2}^{\,2}$ does not possess a conjugate partner in the same sense.

To clarify this point, the conjugation procedure is examined more closely. The tower $T_{\,2}^{\,3}$
can be obtained from $T_{\,2}^{\,1}$ by rotating the horizontally arranged building blocks into a vertical orientation.
The same operation can be applied to the building blocks of the tower $T_{\,2}^{\,2}$: rotating the horizontally
arranged blocks into a vertical position produces a tower that is identical in appearance to the original one,
even though its internal construction differs. This behaviour is a consequence of the square shape of the tower's lateral faces.

The tower $T_{\,2}^{\,2}$ is therefore referred to as a \emph{self-conjugate} tower. As shown in
Proposition~\ref{prop:selfconjugate} below, whenever $n$ is even, the middle tower in the sequence is self-conjugate.

From this perspective, transitions between conjugate towers may be interpreted as transitions between distinct
`geometric phases' with different densities. In contrast, the transition associated with a self-conjugate tower
corresponds to a change between different symmetry realisations of the same geometric state. In this sense,
transitions between conjugate towers resemble first-order-like phase transitions, while transitions involving
self-conjugate towers may be viewed as second-order-like phase transitions.

In the case $n=3$, two pairs of towers are present, namely $(T_{\,3}^{\,1}, T_{\,3}^{\,4})$ and
$(T_{\,3}^{\,2}, T_{\,3}^{\,3})$. Within each pair, a first-order-like geometrical phase transition is possible
through a pure rearrangement of the building blocks. By contrast, for $n=4$ the tower pairs between which such
geometrical phase transitions may occur are $(T_{\,4}^{\,1}, T_{\,4}^{\,5})$ and $(T_{\,4}^{\,2}, T_{\,4}^{\,4})$,
while the tower $T_{\,4}^{\,3}$ is self-conjugate and is associated with a second-order-like phase transition.

These examples already reveal the general structure. For arbitrary $n$, the conjugate tower pairs are
\begin{equation}
(T_{\,n}^{\,1}, T_{\,n}^{\,n+1}), \, 
(T_{\,n}^{\,2}, T_{\,n}^{\,n}), \, 
(T_{\,n}^{\,3}, T_{\,n}^{\,n-1}), \, \dots
\label{eq:conjugatepairs}
\end{equation}

If $n$ is odd, there are $(n+1)/2$ such pairs, and every tower has a conjugate partner. If $n$ is even,
there are $n/2$ conjugate pairs, while a single `middle' tower, $T_{\,n}^{\,n/2+1}$, remains
self-conjugate. This odd--even distinction plays a central role in the classification of geometrical
phase transitions and underlies the general formulae derived in Sec.~\ref{sec:generalvolume}.

The following proposition establishes that the lateral faces of a self-conjugate tower are necessarily square,
as claimed above. This requires the width and height of a lateral face to be equal. Since the width of a tower
face is equal to the length of its building blocks, it is sufficient to demonstrate that this length
coincides with the height of the tower.

\begin{proposition}
\label{prop:selfconjugate}
The lateral faces of a self-conjugate tower are square.
\end{proposition}

\begin{proof}
The equality of the width and height of a lateral face is established as follows.

First, a self-conjugate tower occurs only when $n$ is even and appears at the fragmentation stage $i = n/2 + 1$.

Second, the length of the building blocks at stage $i$ follows from the fragmentation rule. The initial length
is divided by four at the first stage and is subsequently halved at each stage. The block length at stage $i$
is therefore
\begin{equation}
    l_i = \frac{2^{\,n+2}a}{2^{\,i+1}} = 2^{\,n-i+1}a\,.
    \label{eq:block_length}
\end{equation}
Substitution of $i = n/2 + 1$ into Eq.~\eqref{eq:block_length}
yields $l_{\,n/2+1} = 2^{\,n/2}a$.

Third, the height of the tower is determined by the number of 
levels multiplied by $a$. A single level is formed at the first 
stage, and the number of levels doubles at each subsequent stage. 
The number of levels at stage $i$ is thus $2^{\,i-1}$, therefore 
the height of the tower $h_i$ at stage $i$ is
\begin{equation}
    h_i = 2^{\,i-1}a\,.
    \label{eq:tower_height}
\end{equation}
Substitution of $i = n/2 + 1$ into Eq.~\eqref{eq:tower_height} gives
$h_{\,n/2+1} = 2^{\,n/2}a$.

Hence, the width and height of the lateral faces are equal, and 
the lateral faces of a self-conjugate tower are square.
\end{proof}

\begin{corollary}
\label{cor:numtowers}
The number of towers $T(n)$ created from a parent object 
with dimensions $2^{\,n+2}a \times a \times a$ is given by
\begin{equation}
    T(n) = n+1\,.
    \label{eq:numOfTowers}
\end{equation}
\end{corollary}

\begin{proof}
At the final stage of fragmentation, denoted by $f$, the block 
length $l_f$ is equal to $a$. Equation~\eqref{eq:block_length} 
therefore takes the form
\[
    a = 2^{\,n-f+1}a\,.
\]
It then follows directly that the final stage index is $f = n+1$. 
Since towers are produced at stages $i = 1, 2, \dots, f$, the 
total number of towers produced during the fragmentation process 
is equal to $n+1$, which completes the proof.
\end{proof}

For each $n$, the sequence of towers $T_{\,n}^{\,i}$ exhibits 
a natural pairing: $T_{\,n}^{\,1}$ with $T_{\,n}^{\,n+1}$, 
$T_{\,n}^{\,2}$ with $T_{\,n}^{\,n}$, and so on. Towers within 
each pair are constructed from \emph{identical-looking} building 
blocks, but differ by a global rearrangement that leads to 
different enclosed cavity volumes. This pairing is made precise 
in the following corollary.

\begin{corollary}
\label{cor:conjugatepairs}
The conjugate tower pairs are given by
$(T_{\,n}^{\,i}, T_{\,n}^{\,n-i+2})$, where $i = 1,2,\dots,n+1$.
That is, a tower at fragmentation stage $i$ is paired with the tower
at fragmentation stage $n-i+2$.
\end{corollary}

\begin{proof}
If a tower at fragmentation stage $i$ is paired with a tower at
fragmentation stage $p$ (without loss of generality, 
$i \leqslant p$ may be assumed), then the length $l_i$ of the 
building blocks at stage $i$ is equal to the height $h_p$ of 
the tower at the paired stage $p$. Using 
Eqs.~\eqref{eq:block_length} and \eqref{eq:tower_height}, 
this condition can be written as
\[
    2^{\,n-i+1}a = 2^{\,p-1}a\,.
\]
From this relation, the paired fragmentation stage is obtained 
as $p = n-i+2$. Substituting back, the conjugate tower pairs 
are $(T_{\,n}^{\,i}, T_{\,n}^{\,n-i+2})$, which completes 
the proof.
\end{proof}

\begin{corollary}
\label{cor:parity}
The number of conjugate tower pairs generated from a parent
object indexed by $n$ depends on the parity of $n$.
For odd $n$, the $n+1$ towers form $(n+1)/2$ conjugate pairs.
For even $n$, the towers form $n/2$ conjugate pairs, and the 
remaining tower $T_{\,n}^{\,n/2+1}$ is self-conjugate.
\end{corollary}

\begin{proof}
The result follows directly from 
Corollary~\ref{cor:conjugatepairs} by counting the pairs.
\end{proof}

\begin{remark}
The emergence of conjugate and self-conjugate towers is a 
purely geometric consequence of the 
fragmentation\textendash{}reassembly procedure.
The term ``geometric phase'' is used here in an analogical 
sense, to denote distinct reassembled configurations of 
identical building blocks that enclose different volumes.
No energetic or thermodynamic interpretation is implied.
\end{remark}

\subsection{General volume formula}
\label{sec:generalvolume}

Motivated by the explicit examples presented in the preceding 
subsections, a general expression for the maximal volume attained 
at each fragmentation stage is derived.

The volume of a tower is given by the sum of the material volume 
$2^{n+2}a^{3}$ of the parent object, which remains 
unchanged throughout the fragmentation process, and the volume 
of the enclosed space (cavity). The cavity volume is given by 
the product of the cross-sectional area of the cavity and the 
height of the tower, with the latter given by 
Eq.~\eqref{eq:tower_height}. The cross-sectional area of the 
cavity is equal to the square of the building-block length $l_i$, 
given by Eq.~\eqref{eq:block_length}. The enclosed cavity volume 
is therefore expressed as
\[
    l_i^{2} h_i = 2^{\,2n-(i-1)} a^{3}.
\]
The total tower volume then follows immediately.

\begin{proposition}
\label{prop:volume}
The volume $V_{\,n}^{\,i}$ of the tower $T_{\,n}^{\,i}$ 
reassembled at fragmentation stage $i$ from a parent 
object of dimensions $2^{\,n+2}a \times a \times a$ is given by
\begin{equation}
    V_{\,n}^{\,i} = 2^{\,n+2} a^{3} + 2^{\,2n-(i-1)} a^{3},
    \qquad i = 1,2,\dots,n+1.
    \label{eq:Vn_i}
\end{equation}
\end{proposition}

The corresponding relative volume is obtained as
\begin{equation}
    R_{\,n}^{\,i} = \frac{V_{\,n}^{\,i}}{V_{\,n}^{\,0}}
    = \frac{2^{\,n+2}a^3 + 2^{\,2n-(i-1)}a^3}{2^{\,n+2}a^3}
    = 1 + 2^{\,n-i-1}.
    \label{eq:general_relative_volume}
\end{equation}

The variation of the relative volume as a function of the 
fragmentation stage $i$ for several values of $n$ is illustrated 
in Fig.~\ref{fig:rv}. As established in 
Corollary~\ref{cor:numtowers}, the final fragmentation occurs 
at stage $i=n+1$. The relative volume at the final stage is 
therefore independent of $n$ and is given by
\begin{equation}
    R_{\,n}^{\,n+1} = 1 + 2^{-2} = \frac{5}{4}.
    \label{eq:5/4}
\end{equation}

The value $5/4$ thus represents a geometric upper bound for the 
relative volume obtained after complete fragmentation within 
the present model.

The plots further indicate that the first fragmentation stage 
produces the largest increase in volume, after which the relative 
volume decreases monotonically with increasing stage index. This 
behaviour is governed by the evolution of the central cavity: 
the parent object encloses no void space, whereas the first 
fragmentation generates the largest possible cavity. At each 
subsequent stage, the cross-sectional area of the cavity is 
reduced by a factor of four, while the height of the tower 
increases only by a factor of two. Consequently, the enclosed 
cavity volume\,---\,and hence the cavity contribution to the relative 
volume\,---\,is reduced by a factor of two at each stage. The 
relative volume itself decreases monotonically but does not 
halve, owing to the additive contribution of the material volume.

\begin{figure}[t]
\centering
\begin{tikzpicture}[scale=1.2]
    % Draw axes with thicker lines
    \draw[line width=0.8pt, ->] (0,0) -- (5.5,0) node[right] {\large $i$};
    \draw[line width=0.8pt, ->] (0,0) -- (0,5.5) node[above] {\large $R_n^i$};
    
    % Draw grid
    \draw[gray!30, very thin] (0,0) grid (5,5);
    
    % X-axis tick marks and labels
    \foreach \x in {0,1,2,3,4,5} {
        \draw[line width=0.8pt] (\x,0.05) -- (\x,-0.05);
        \node[below] at (\x,-0.1) {\large \x};
    }
    
    % Y-axis tick marks and labels
    \foreach \y in {0,1,2,3,4,5} {
        \draw[line width=0.8pt] (0.05,\y) -- (-0.05,\y);
        \node[left] at (-0.1,\y) {\large \y};
    }
    
    % Plot n=4 (solid, thick)
    \draw[black, line width=1.0pt] 
        (0,1) -- (1,5) -- (2,3) node[sloped, pos=0.5, above] {\large $n=4$}  -- (3,2) -- (4,1.5) -- (5,1.25);
    
    % Plot n=3 (solid, thick)
    \draw[black, line width=1.0pt] 
        (0,1) -- (1,3) -- (2,2) node[sloped, pos=0.5, above] {\large $n=3$} -- (3,1.5) -- (4,1.25);
    
    % Plot n=2 (solid, thick)
    \draw[black, line width=1.0pt] 
        (0,1) -- (1,2) -- (2,1.5) node[sloped, pos=0.5, above] {\large $n=2$} -- (3,1.25);
    
    % Plot n=1 (solid, thick)
    \draw[black, line width=1.0pt] 
        (0,1) -- (1,1.5) -- (2,1.25);
    \node[black] at (2.0,1.0) {\large $n=1$};
    
    % Plot n=0 (solid, thick)
    \draw[black, line width=1.0pt] 
        (0,1) -- (1,1.25);
    \node[black] at (1.0,0.75) {\large $n=0$};
    
\end{tikzpicture}
\caption{Relative volume $R_{\,n}^{\,i}$ as a function of fragmentation 
stage $i$ for several values of $n$. Each curve corresponds to a fixed 
parent object length $l=2^{n+2}a$. The first fragmentation stage yields 
the maximal relative volume, followed by a monotonic decrease towards 
the limiting value $R=5/4$ at $i=n+1$.}
\label{fig:rv}
\end{figure}

\subsection{Remarks on paired configurations and ``phase'' interpretation}

The pairing structure is closely related to the symmetry properties of the tower configurations.

The relative-volume difference for a conjugate pair $(i,n+2-i)$ is given by
\begin{align}
    \Delta R_{\,n}^{\,i}
    & = R_{\,n}^{\,i} - R_{\,n}^{\,n+2-i} \nonumber \\
    & = \bigl(1+2^{\,n-i-1}\bigr) - \bigl(1+2^{\,n-(n+2-i)-1}\bigr)\nonumber\\
    & = 2^{\,n-i-1} - 2^{\,i-3}.
    \label{eq:DeltaR}
\end{align}
This difference is maximal for $i=1$ (corresponding to the pairing of the first and last towers)
and decreases as the index $i$ approaches the central value, $i = n/2+1$.

The relative-volume difference vanishes when
\begin{equation}
    i = \frac{n}{2} +1\,.
    \label{eq:n/2+1}
\end{equation}
Since $i$ must be a positive integer, this condition requires careful interpretation.
If $n$ is even, Eq.~\eqref{eq:n/2+1} yields an integer value, and the pair $(i,n+2-i)$ reduces to
$(n/2+1,n/2+1)$. In this case, two conjugate towers coincide and form a single self-conjugate tower.
If $n$ is odd, such a reduction does not occur, and a self-conjugate tower never exists.

Conjugate towers are characterized by a rearrangement of \emph{identical-looking} building blocks.
When these towers coincide in the even-$n$ case, the resulting self-conjugate tower inherits properties
of both configurations. It can therefore be inferred that the self-conjugate tower may be constructed
from the same building blocks in two distinct orientations, implying the presence of an internal 
symmetry, in agreement with the behaviour already observed in Sec.~\ref{sec:casen2}.

These paired configurations provide a simple geometric analogue of distinct phases. While they are
composed of \emph{identical-looking} building blocks, they differ by a global rearrangement that
modifies the enclosed cavity volume. Rearrangement-induced transitions between conjugate towers
thus correspond to transitions between geometrically distinct configurations with different effective
densities.

The geometric transitions observed in the present model admit a natural interpretation in the language
of phase transitions. The relative volume $R_{\,n}^{\,i}$ plays the role of a macroscopic order parameter
distinguishing tower configurations. Transitions between conjugate towers involve discontinuous changes in
$R_{\,n}^{\,i}$ and may therefore be regarded as first-order-like.

In contrast, when $n$ is even, the middle tower is self-conjugate: rearrangement of the building blocks
alters their internal orientation without changing the overall geometry or volume. This transition is
associated with a discrete rotational symmetry and corresponds to a change in internal symmetry
rather than in macroscopic structure. In this sense, the self-conjugate tower represents a geometric
analogue of a second-order phase transition.

In summary, the fragmented tower configurations are organized into conjugate pairs distinguished
by a discrete symmetry of relative volume, with a self-conjugate configuration arising only for even 
$n$. This structure provides a geometric analogue of phase behavior, in which conjugate towers represent
distinct macroscopic phases, while the self-conjugate tower corresponds to a symmetry-driven
transition without a change in density.

\section{Connection to granular systems}
\label{sec:granular_connection}

The connection between the present model and real granular systems
is established by employing Eq.~\eqref{eq:general_relative_volume}
to evaluate the relative volume accessible to grains in experimental
systems, while Eq.~\eqref{eq:DeltaR} is employed to identify the
geometric phase transitions predicted by the model. Both expressions
are formulated in terms of the discrete parameters $n$ and $i$. To
enable experimental comparison, these parameters must be re-expressed
in terms of measurable grain characteristics.

Throughout this section, it is assumed that grains do not interact with one another,
that frictional effects are negligible, and that individual grains may be regarded as
effectively weightless, such that no load is transmitted through the assembly. These assumptions
are adopted in order to isolate the purely geometric contributions to volume evolution
that are central to the present model.

\subsection{Remarks on experimental observation of the relative volume}
\label{sec:experimental_observation}

A granular sample composed of identical grains with square 
cross-section of side length $a$ is considered. It is further 
assumed that these grains may be regarded as having been produced 
by sequential slicing of a hypothetical elongated parent object, 
in accordance with the fragmentation rule introduced earlier.

By comparing the total mass of the granular sample with the mass 
of an individual grain, the total number of grains in the sample, 
denoted by $N_g$, is determined. Let $l_g$ denote the length of 
an individual grain. The length $l$ of the hypothetical parent 
object is then given by
\begin{equation}
    l = N_g \, l_g\,.
    \label{eq:initial_length}
\end{equation}
Within the model, the same length is expressed by 
Eq.~\eqref{eq:Init_length_model} in terms of the discrete 
parameter $n$. By equating the two expressions and solving 
for $n$, one obtains
\begin{equation}
    n = \log_{2}\left(\frac{N_g \, l_g}{a}\right) - 2\,.
    \label{eq:n}
\end{equation}
Here, the parameter $a$ is assigned the same meaning in both 
the model description and the experimental system.

The next step is to identify the fragmentation stage of the 
hypothetical parent object at which the size of a model grain 
coincides with that of the grains used in the experiment. At 
fragmentation stage $i$, the grain length $l_i$ (given by 
Eq.~\eqref{eq:block_length}) must match the experimental grain 
length $l_g$, so that
\[
    2^{\,n-i+1} a = l_g\,.
\]
Using Eq.~\eqref{eq:n}, this condition yields
\begin{equation}
    i = \log_{2} N_g - 1\,.
    \label{eq:i_fragmentation}
\end{equation}
As a consistency check, for a single grain ($N_g = 1$), 
Eq.~\eqref{eq:i_fragmentation} yields $i = 0$, corresponding 
to the absence of fragmentation, as expected.

Substituting Eqs.~\eqref{eq:n} and \eqref{eq:i_fragmentation} 
into Eq.~\eqref{eq:general_relative_volume} yields
\begin{equation}
    R = 1 + \frac{l_g}{4a}\,.
    \label{eq:R}
\end{equation}
This result demonstrates that the relative volume is determined 
exclusively by the geometric parameters of the grains and is 
independent of the details of the fragmentation history. All 
dependence on the hypothetical parent object is eliminated, 
rendering the indices $n$ and $i$ redundant; they are therefore 
omitted from Eq.~\eqref{eq:R}.

Although the relative volume $R$ is expressed in Eq.~\eqref{eq:R} as a single function of
the geometric ratio $l_g/a$, this quantity should be understood as a characteristic upper
bound rather than as a uniquely reproducible experimental outcome. Under nominally
identical preparation conditions, repeated realizations of the system are expected
to yield a distribution of measured relative volumes clustered below this bound.

For a granular assembly composed of grains with square cross section, the relative volume
may therefore be determined experimentally by measuring only the grain length $l_g$ and
the side length $a$ of the cross section. Upon successive fragmentation, these measurements
may be repeated until the terminal stage is reached, at which $l_g = a$. In this limit,
the relative volume, given by Eq.~\eqref{eq:R}, approaches the value $5/4$,
in agreement with Eq.~\eqref{eq:5/4}.

Exact quantitative agreement with the theoretical curve is not expected under experimental
conditions. Nevertheless, the measured values are expected to follow the qualitative trend
predicted by Eq.~\eqref{eq:R} and to remain bounded above by the maximal-volume curve
provided by the model. Assuming square cross sections, these considerations may be summarised
as follows.

\begin{theorem}[Liza's theorem]
\label{thm:aspect_ratio}
If a grain has length $l_g$ and a square cross section with 
side length $a$, the upper bound of the relative volume is 
determined solely by the ratio $l_g/a$ according to 
Eq.~\eqref{eq:R}.
\end{theorem}

\begin{remark}[Liza's limit]
\label{rem:limiting_value}
The upper bound of the relative volume decreases monotonically 
as the grain length $l_g$ decreases. Its absolute minimum is 
attained when $l_g = a$, at which point the upper bound equals 
$5/4$, a universal value referred to as Liza’s limit or Liza’s number.
The upper bound is independent of the number of grains and always exceeds unity,
corresponding to the volume obtained by packing all grains adjacently without voids.
\end{remark}

Equation~\eqref{eq:R} may be tested against existing experimental data. In Ref.~\cite{villarruel2000},
experiments were conducted on monodisperse nylon 6/6 cylindrical rods with diameter
$1.8\,\mathrm{mm}$ and length $7.0\,\mathrm{mm}$. Approximately 7200 rods were poured into a vertical
glass tube of height $1\,\mathrm{m}$ and inner diameter $1.90\,\mathrm{cm}$. Although the present
model is formulated for grains with square cross section, a reasonable estimate is expected to be
obtained for cylindrical rods of comparable transverse dimensions. The rods are therefore
approximated as grains with square cross section of side length $a = 1.8\,\mathrm{mm}$ and
grain length $l_g = 7.0\,\mathrm{mm}$. Substitution into Eq.~\eqref{eq:R} then yields
\[
    R = 1 + \frac{l_g}{4a} \simeq 1.97\,.
\]

Experimentally, the quantity reported is the \emph{packing fraction} $\rho$ (also commonly
denoted $\phi$), defined as the fraction of the available volume occupied by the particles.
When the longest grain dimension is small compared to the characteristic dimensions of the
container\,---\,in the present case, the rod length relative to the tube diameter\,---\,the packing
fraction may be approximated as $\rho \simeq R^{-1}$. This estimate gives $\rho \simeq 0.51$,
which lies within the experimentally observed range $\rho = 0.49$--$0.55$
reported in Ref.~\cite{villarruel2000}.

\subsection{Comparison with experimental scaling laws}

Experimental measurements have shown that the packing fraction 
$\phi$ of randomly oriented elongated particles decreases with 
increasing aspect ratio according to an asymptotic scaling 
law~\cite{philipse1996,williams2003}
\begin{equation}
2\phi \times \alpha = \langle \gamma \rangle, 
\qquad \alpha \gg 1,
\label{eq:philipse_scaling}
\end{equation}
where $\alpha = L/D$ is the aspect ratio, $L$ is the rod length, 
$D$ is the rod diameter, and $\langle \gamma \rangle$ is the 
average number of contacts per particle. Experiments indicate 
$\langle \gamma \rangle \approx 10.8$ for randomly packed 
rods~\cite{williams2003}.

The present geometric model can be compared with this scaling 
in the limit of large aspect ratio. For a grain with length 
$l_g$ and square cross-section of side $a$, the aspect ratio 
is $\alpha = l_g/a$. The packing fraction is related to the 
relative volume by $\phi \simeq 1/R$. Substituting 
Eq.~\eqref{eq:R} yields
\begin{equation}
\phi \simeq \frac{1}{R} = \frac{1}{1 + l_g/(4a)} 
    = \frac{4}{4 + \alpha}.
\label{eq:phi_alpha}
\end{equation}

For large aspect ratios ($\alpha \gg 1$), this reduces to
\begin{equation}
\phi \approx \frac{4}{\alpha}, 
\qquad \text{or equivalently} \qquad 
2\phi \times \alpha \approx 8.
\label{eq:scaling_model}
\end{equation}

This result demonstrates that the present purely geometric 
model captures the correct $\phi \propto 1/\alpha$ scaling 
observed experimentally. The numerical coefficient differs 
from the experimental value ($8$ versus $\approx 10.8$), which 
is expected given that the model provides an upper bound on 
accessible volume and assumes maximal-volume configurations 
rather than typical random packings. The agreement in scaling behavior confirms that volume evolution 
under fragmentation, as described by the present model, is 
consistent with the experimental observation that packing 
densities are size-invariant and have a purely geometric 
origin~\cite{williams2003}.

Figure~\ref{fig:experimental_comparison} compares the model 
predictions with experimental data from Freeman et al.~\cite{freeman2019} 
for aspect ratios ranging from $\alpha = 4$ to $\alpha = 32$. 
The experimental packing fractions lie consistently above the 
theoretical curve, as expected for a model that provides an upper 
bound on relative volume (or equivalently, a lower bound on packing 
fraction). The agreement in scaling behavior ($\phi \propto 1/\alpha$) 
validates the geometric approach.

\begin{figure}[t]
\centering
\begin{tikzpicture}[scale=0.8, font=\large]
    % Axes
    \draw[line width=0.8pt, ->] (0,0) -- (8.5,0) node[right] {$\alpha$};
%    \draw[line width=0.8pt, ->] (0,0) -- (5.5,0) node[right] {\large $i$};
    \draw[line width=0.8pt, ->] (0,0) -- (0,6.4) node[above] {$\phi$};
    
    % Axis labels
    \foreach \x in {0,5,10,15,20,25,30}
        \draw (\x/4,0.05) -- (\x/4,-0.05) node[below] {\x};
    \foreach \y in {0,0.1,0.2,0.3,0.4,0.5,0.6}
        \draw (0.05,\y*10) -- (-0.05,\y*10) node[left] {\y};
    
    % Model prediction: phi = 4/(4+alpha)
    \draw[black, thick, domain=2.6:32, samples=100] 
        plot (\x/4, {40/(4+\x)});
        
    % Asymptotic scaling law: phi = 5.4/alpha
    \draw[black, thick, dashed, domain=9.0:32, samples=100] 
        plot (\x/4, {54/\x});
    
    % Experimental data points with error bars (downward only)
    \filldraw[black] (1,5.68) circle (3pt);
    \draw[black] (1,5.06) -- (1,5.68);
    
    \filldraw[black] (2,4.55) circle (3pt);
    \draw[black] (2,4.04) -- (2,4.55);
    
    \filldraw[black] (2.5,4.67) circle (3pt);
    \draw[black] (2.5,4.01) -- (2.5,4.67);
    
    \filldraw[black] (4,3.17) circle (3pt);
    \draw[black] (4,2.42) -- (4,3.17);
    
    \filldraw[black] (5,2.05) circle (3pt);
    \draw[black] (5,1.48) -- (5,2.05);
    
    \filldraw[black] (6,1.53) circle (3pt);
    \draw[black] (6,1.09) -- (6,1.53);
    
    \filldraw[black] (8,1.32) circle (3pt);
\end{tikzpicture}
\caption{Experimental packing fraction $\phi_\infty$ as a function of 
aspect ratio $\alpha$ (circles with downward error bars) from Freeman et 
al.~\cite{freeman2019} compared with the present model 
(Eq.~\eqref{eq:phi_alpha}, solid line) and the asymptotic 
scaling law (Eq.~\eqref{eq:philipse_scaling}, dashed line). 
For $\alpha = 32$, the error bar is smaller than the symbol size and 
not visible. The model provides a lower bound on packing fraction, 
while the asymptotic scaling law provides an empirical upper bound 
based on contact statistics.}
\label{fig:experimental_comparison}
\end{figure}

The asymptotic scaling law (Eq.~\eqref{eq:philipse_scaling}), shown 
as a dashed line in Fig.~\ref{fig:experimental_comparison}, provides 
an empirical upper bound on packing fraction based on contact 
statistics in random packings. The experimental data lie between 
the two theoretical curves: above the geometric lower bound from 
the present model and below (or near) the empirical upper bound from 
the asymptotic scaling law. This positioning is physically reasonable, 
as real random packings achieve packing fractions intermediate between 
the maximal-volume geometric configurations considered here and the 
contact-limited configurations reflected in the asymptotic scaling. 
The convergence of both theoretical curves at large $\alpha$, 
with experimental data consistently confined between them, 
confirms that geometric packing constraints dominate throughout 
the range of aspect ratios studied, with the asymptotic regime 
exhibiting universal scaling behavior.

\subsection{Remarks on experimental observation of geometric phase transitions}

An important question remains: can the geometric, rearrangement-induced 
phase transitions described above be observed experimentally? In 
addressing this question, two physical constraints must be taken into 
account. First, once a grain has been sliced, cut, or crushed, the 
resulting fragments cannot spontaneously reattach to reconstitute the 
original solid grain. Second, the levels of a tower are not mechanically 
bonded to one another; the upper levels merely \emph{rest} on the lower 
ones under gravity.

A tower at fragmentation stage $i$ is constructed from grains of length
$l_i = 2^{\,n-i+1}a$ and has a height $h_i = 2^{\,i-1}a$ (see 
Eqs.~\eqref{eq:block_length} and~\eqref{eq:tower_height}). The conjugate 
tower at stage $n-i+2$ is instead constructed from grains of length 
$l_{n-i+2} = 2^{\,i-1}a$ and has a height $h_{n-i+2} = 2^{\,n-i+1}a$.

At first glance, the two towers appear to be assembled from 
\emph{identical grains}: in the former case the grains are arranged 
horizontally, whereas in the latter they are aligned vertically. 
Crucially, however, this apparent equivalence is misleading. The 
vertically oriented ``grain'' is not a single object, but rather a 
stack of smaller fragments, whereas the horizontally oriented grain 
is an intact solid piece.

For illustration, the case $n=2$, shown in Fig.~\ref{fig:n2}, may be 
considered. At fragmentation stage $i=1$, the grain length is 
$l_1 = 2^{\,2-1+1}a = 4a$, and the tower consists of a single level 
of height $h_1 = 2^{\,1-1}a = a$. At the conjugate stage $n-i+2 = 3$, 
the grain length is $l_3 = 2^{\,i-1}a = a$, and the tower height is
$h_3 = 2^{\,2-1+1}a = 4a$. Four unit blocks are stacked vertically, 
producing an object that \emph{appears} identical, at a coarse-grained 
geometric level, to a single grain of length $4a$ at stage $i=1$.

As discussed, a tower at fragmentation stage $i$ is paired 
with a tower at stage $n-i+2$. This pairing reflects the fact that 
the two configurations are composed of building blocks that are 
geometrically indistinguishable at the level of their external shape, 
despite differing in their internal structure. Without loss of 
generality, the case $i < n-i+2$ is considered. In this situation, 
the grains at stage $n-i+2$ are smaller than those at stage $i$. 
Although their vertical alignment reproduces the geometry of the 
larger grains, recombination into the intact grains of stage $i$ 
is not possible. Moreover, since the vertically stacked fragments 
are not bonded to one another, any attempt to rearrange them would 
cause the stack to collapse, preventing reconstruction of the tower 
configuration at stage $i$. By contrast, the grains at stage $i$ 
\emph{can} be rearranged to form the tower corresponding to stage 
$n-i+2$.

It is therefore concluded that the geometric phase transition is 
intrinsically directional and may occur only in the direction 
$i \to n-i+2$. However, if the intact grains at stage $i$ happen 
to be oriented vertically\,---\,matching the orientation of the 
stacked fragments at stage $n-i+2$\,---\,then rotation to horizontal 
alignment enables the reverse transition $n-i+2 \to i$ without 
requiring recombination of fragments.

With this constraint in mind, the following experimental situation may 
be considered. Suppose that $N_g$ grains, each of length $l_g$, are 
given. The fragmentation stage $i$ of the hypothetical parent object 
corresponding to this configuration is determined by 
Eq.~\eqref{eq:i_fragmentation}. The conjugate stage 
$i_c \equiv n-i+2$, expressed directly in terms of experimentally 
measurable parameters, is then given by
\begin{equation}
    i_c = \log_{2}\!\left(\frac{l_g}{a}\right) + 1.
    \label{eq:ic_fragmentation}
\end{equation}

If
\begin{equation}
    i < i_c,
    \label{eq:Condition4PhaseTransition}
\end{equation}
the occurrence of a geometric phase transition is permitted; otherwise, 
such a transition is prohibited.

By substituting the expressions for $n$ and $i$ from 
Eqs.~\eqref{eq:n} and \eqref{eq:i_fragmentation} into 
Eq.~\eqref{eq:DeltaR} for the difference between the relative volumes 
of conjugate configurations, the following expression is obtained:
\begin{equation}
    \Delta R = \frac{l_g}{4a} - \frac{1}{16}N_g.
    \label{eq:DeltaR_Obser}
\end{equation}
Here, the indices on $\Delta R$ are redundant and have therefore been 
omitted.

As already assumed above, the comparison is made under the convention 
$i < n-i+2$, adopted without loss of generality. Since the relative 
volume decreases monotonically with increasing fragmentation stage, 
the difference $\Delta R$ is therefore non-negative. 
Equation~\eqref{eq:DeltaR_Obser} then yields the criterion
\begin{equation}
    \frac{N_g a}{4\,l_g} \leqslant 1.
    \label{eq:Criteria4PhaseTransition}
\end{equation}
Equality corresponds to the self-conjugate configuration, in which 
the transition involves a change in internal geometric symmetry 
without any change in the enclosed volume.

The phase-transition criterion~\eqref{eq:Criteria4PhaseTransition} 
provides insight into the interpretation of Eq.~\eqref{eq:R}. 
As discussed above, geometric phase transitions are directional: 
they can occur only from intact grains to stacked configurations 
($i \to n-i+2$ with $i < (n+1)/2$), not in reverse, since stacked 
fragments would collapse during rearrangement. When the grain-number 
criterion is not satisfied\,---\,that is, when $i > (n+1)/2$\,---\,only 
a single accessible configuration exists, and repeated experimental 
measurements of the relative volume $R$ are expected to form a single 
cluster centred around the characteristic value prescribed by 
Eq.~\eqref{eq:R}. If the criterion is satisfied, meaning 
$i < (n+1)/2$, two conjugate geometric configurations become 
accessible. In this regime, repeated preparations may populate 
either configuration, resulting in two distinct clusters in the 
measured values of $R$, separated by $\Delta R$ 
(Eq.~\eqref{eq:DeltaR_Obser}) and corresponding to the two 
conjugate geometric phases.

If condition~\eqref{eq:Criteria4PhaseTransition} is not satisfied, 
the number of grains $N_g$ may, in principle, be reduced (provided 
$N_g \gg 1$) until the criterion is met. Combining these requirements 
yields
\[
    1 \ll N_g \leqslant \frac{4l_g}{a},
\]
indicating that the occurrence of a geometric phase transition requires 
grains with a sufficiently large aspect ratio.

The implications of this criterion may be illustrated by concrete 
examples. For the nylon rods discussed in 
Sec.~\ref{sec:experimental_observation} (with $l_g = 7.0\,\mathrm{mm}$ 
and $a = 1.8\,\mathrm{mm}$), the phase-transition condition implies 
$N_g < 16$. As a further example with larger aspect ratio, matchsticks 
with tips removed, characterised by a cross-sectional side length 
$a = 2\,\mathrm{mm}$ and length $l_g = 5\,\mathrm{cm}$, yield the 
condition $N_g < 100$.

In both cases, the number of grains required for the observation of 
the transition is relatively small. It is therefore concluded that 
geometric phase transitions of the type described here are restricted 
to mesoscopic systems, which may explain why they are not commonly 
observed in everyday granular materials.

Experimental support for the local nature of these transitions may 
be found in X-ray tomography studies. Freeman et al.~\cite{freeman2019} 
observed local alignment of small groups of rods (forming what may be 
termed \emph{domains}) in containers filled with grains of aspect ratio 
$\alpha = 8$. For these rods in containers of diameter $\approx 4.9\,l_g$, 
the criterion~\eqref{eq:Criteria4PhaseTransition} requires $N_g < 32$ 
for geometric phase transitions to occur. The observed local domains, 
while themselves randomly oriented within the container interior, are 
consistent with geometric rearrangements occurring in small grain 
clusters rather than uniformly throughout the sample. Notably, X-ray 
tomography reveals that domains near the container wall exhibit 
preferential alignment parallel to the boundary, suggesting that 
external constraints influence the orientation of geometric phases. 
In the container interior, where wall effects are negligible and 
gravity is the dominant external field, domain orientations appear 
randomly distributed. Similarly, Zakine et al.~\cite{zakine2025} 
found that pistachio shells\,---\,whose shapes deviate significantly 
from the idealised grains considered here\,---\,compact and interlock 
under mechanical agitation, suggesting that such transitions are 
facilitated by local rearrangements within small clusters.

This naturally motivates consideration of domains whose grain number 
and aspect ratio satisfy the phase-transition criterion derived above. 
Local rearrangements of this type may then cumulatively contribute to 
the macroscopic volume reduction observed during agitation. To test 
this hypothesis experimentally, the bulk packing density may be 
modelled as a superposition of contributions from different geometric 
configurations (or phases). It should be noted that while the 
criterion~\eqref{eq:Criteria4PhaseTransition} establishes an upper 
bound on domain size ($N_g \lesssim 4l_g/a$), it does not determine 
the actual distribution of grain numbers within domains. The statistics 
of domain sizes in bulk granular assemblies\,---\,including the most 
probable domain size and the distribution of cluster populations\,---\,remain 
open questions that would require a dynamical theory to address.

Although the present model is purely geometric and entirely static in 
nature, any experimental realisation necessarily involves real grains 
subject to friction, gravity, and contact constraints. Gentle 
mechanical agitation is therefore introduced not to endow the system 
with thermodynamic character, but simply to facilitate internal 
rearrangements and allow the system to explore its admissible 
geometric configurations. In this sense, agitation serves as a 
practical means of reducing kinetic trapping and history dependence, 
rather than as a source of thermalisation or energetic driving.

\paragraph{Experimental signatures.}
The geometric phase transition mechanism admits two distinct 
experimental regimes depending on the total grain number.

\textit{Small-system regime} ($N_g \lesssim 4l_g/a$): When the 
criterion~\eqref{eq:Criteria4PhaseTransition} is satisfied for 
the entire grain assembly, repeated preparations under controlled 
agitation should yield a bimodal distribution of bulk relative 
volumes $R$, with peaks separated by $\Delta R$ 
(Eq.~\eqref{eq:DeltaR_Obser}) corresponding to global conjugate 
configurations.

\textit{Large-system regime} ($N_g \gg 4l_g/a$): When the total 
grain number exceeds the criterion value, global phase transitions 
are prohibited, but local transitions within domains containing 
$N_g \lesssim 4l_g/a$ grains remain possible. In this regime, 
the bulk relative volume forms a single distribution, but X-ray 
tomography should reveal aligned grain clusters (domains) with 
characteristic maximum size $N_{\text{domain}} \sim 4l_g/a$. 
Since this constraint scales linearly with the aspect ratio 
$\alpha = l_g/a$, experiments conducted on grains of systematically 
varying aspect ratios should reveal that characteristic maximum 
domain sizes increase proportionally with $\alpha$. Observation 
of such scaling would provide direct experimental evidence for 
the geometric phase transition mechanism, while its absence would 
indicate alternative origins for the local alignment patterns.

\section{Summary}

A geometric model of granular matter composed of grains with square 
cross section has been developed by representing the system as 
originating from a hypothetical elongated parent object. The 
characteristic size of this object is fixed by the number and 
dimensions of the constituent grains.

Fragmentation is implemented through a well-defined slicing algorithm, 
applied iteratively and described in detail in the main text. At each 
stage of fragmentation, the resulting fragments are reassembled into 
configurations of maximal admissible volume, referred to as 
\emph{towers}. The volumes of these towers were analysed as functions 
of the fragmentation stage. In realistic granular systems composed of 
elongated grains, the constituents at a given stage of fragmentation 
are expected to arrange in a largely unconstrained manner, resulting 
in occupied volumes that lie between those of compact packing and the 
largest geometrically admissible configurations. The construction of 
maximal-volume towers therefore provides a purely geometric upper 
bound on occupied volume, or equivalently, a lower bound on packing 
fraction, and allows the characteristic scale and qualitative trends 
of volume changes induced by fragmentation to be estimated.

It has been shown that, at the initial stage of fragmentation, the 
occupied volume increases relative to that of the parent object. 
With further fragmentation, the volume then decreases monotonically, 
approaching at the final stage the limiting value $5/4$ of the original 
volume. It follows that, throughout the fragmentation process, a 
material composed of elongated grains typically occupies a volume 
larger than that of the original object, unless the grains are packed 
perfectly without voids.

The inverse of the relative volume, corresponding to the packing 
fraction, was shown to obey the asymptotic scaling law 
$\phi \propto 1/\alpha$ for large aspect ratios, in agreement with 
experimental observations. Comparison with experimental data from 
Freeman et al.~\cite{freeman2019} demonstrates that the model provides a geometric lower 
bound on packing fraction, with experimental values lying consistently 
above the theoretical predictions, as expected for realistic random 
packings.

Within the fragmentation sequence, pairs of configurations have been 
identified and termed \emph{conjugate towers}. These towers are 
composed of geometrically indistinguishable building blocks but enclose 
different volumes. When the number of fragmentation stages is even, a 
\emph{self-conjugate} tower appears at the centre of the sequence. The 
magnitude of the volume difference between conjugate towers depends on 
their separation along the fragmentation sequence: towers corresponding 
to nearby stages exhibit smaller differences, whereas the largest 
difference occurs between the first and final stages of fragmentation.

Transitions between conjugate towers may be interpreted as geometric 
phase transitions. A criterion for the observability of such transitions 
has been derived, showing that they are restricted to small systems 
satisfying $N_g \lesssim 4l_g/a$. In large bulk systems where this 
criterion is not globally satisfied, geometric phase transitions may 
nevertheless occur locally within small domains of aligned grains. 
X-ray tomography studies of rod packings reveal such local alignment 
patterns, with characteristic domain sizes of order 
$N_{\text{domain}} \lesssim 4l_g/a$ grains, consistent with the 
theoretical prediction. A directly testable experimental signature 
is proposed: maximum domain size should scale linearly with the 
aspect ratio $\alpha$, providing a quantitative test of the 
geometric phase transition mechanism.

In summary, an analytically tractable geometric model of fragmentation 
and reassembly has been presented, revealing features of volume evolution 
and phase-like behaviour in granular systems. Despite its simplicity, 
the model captures essential geometric aspects of granular matter, 
provides quantitative predictions in agreement with experimental scaling 
laws, and offers a framework for investigating packing geometry, porosity, 
metastable configurations, and local ordering phenomena in elongated 
grain assemblies. The work also identifies the need for a dynamical 
theory to address the statistics of domain size distributions and the 
mechanisms governing cluster formation.

\section*{Acknowledgements}

The author thanks T. B\"orzs\"onyi for valuable comments on the role 
of Van der Waals forces in fine powder packing and for helpful 
discussions during the preparation of this manuscript.

The author is grateful to A. Zaccone for his interest in this work 
and for suggesting the review article Ref.~\cite{zaccone2025}.

The author also thanks E. R. Weeks for kindly providing numerical 
data from the experiments reported in Ref.~\cite{freeman2019}.

This work is dedicated to the memory of the author's grandmother,
Liza Tsertsvadze. Many years ago, she remarked that a given mass
of corn kernels occupies a smaller volume than the flour obtained
from them. This simple observation provided the original motivation
for the present study and inspired the investigation of the geometric
origins underlying such volume changes.

The author acknowledges the use of AI language models (ChatGPT and 
Claude) to assist with manuscript preparation, including text editing, 
LaTeX formatting, and literature organization. All scientific content, 
theoretical development, and conclusions are solely the work of the author.

\bibliographystyle{apsrev4-2}
\bibliography{refs1}

\end{document}